\documentclass[prl,aps,twocolumn,showpacs]{revtex4-1} 
\usepackage{amsmath,amssymb,amsthm}
\usepackage{graphicx}
\usepackage{subfigure}
\usepackage{amsmath}
\usepackage{amsfonts,amssymb}
\usepackage{bm}
\usepackage{float}
\usepackage{color}
\usepackage{sidecap}
\allowdisplaybreaks
\usepackage{soul}
\setstcolor{red}
\usepackage[normalem]{ulem}
\usepackage{hyperref}
\hypersetup{colorlinks=true,breaklinks,urlcolor=blue,linkcolor=blue,citecolor=blue}

\newtheorem{theorem}{Theorem}[section]

\newtheorem{lemma}[theorem]{Lemma}

\begin{document}

\title{Topological Phases in Non-Hermitian Aubry-Andr\'e-Harper Models}
\author{Qi-Bo Zeng}
\author{Yan-Bin Yang}
\author{Yong Xu}
\email{yongxuphy@tsinghua.edu.cn}
\affiliation{Center for Quantum Information, IIIS, Tsinghua University, Beijing, 100084, P. R. China}

\begin{abstract}
Topological phases have recently witnessed a rapid progress in non-Hermitian systems. Here we
study a one-dimensional non-Hermitian Aubry-Andr\'e-Harper model with imaginary periodic or quasiperiodic modulations.
We demonstrate that the non-Hermitian off-diagonal
AAH models can host zero-energy modes at the edges. In contrast to the Hermitian case, the zero-energy mode can be
localized only at one edge. Such a topological phase corresponds to the existence of a quarter winding number defined
by eigenenergy in momentum space. We further find the coexistence of a zero-energy mode located only at one edge and
topological nonzero energy edge modes characterized by a generalized Bott index. In the incommensurate case, a topological non-Hermitian quasicrystal is
predicted where all bulk states and two topological edge states are localized at one edge. Such topological edge modes
are protected by the generalized Bott index.
Finally, we propose an experimental scheme to realize these non-Hermitian models in electric circuits. Our findings add a new direction for exploring topological properties in Aubry-Andr\'e-Harper models.
\end{abstract}

\maketitle
\date{today}

Topological phases have become one of the most fascinating and rapidly developing research field in condensed matter physics in the past decade, both theoretically and experimentally \cite{Hasan, Qi, Armitage}. Despite being found in Hermitian
systems, topological phases have recently sparked tremendous interests in non-Hermitian systems~\cite{Rudner1,Esaki,Bardyn,Poshakinskiy,Zeuner,Malzard,Rudner2,Aguado2016SR,Lee3,Molina,Joglekar2016PRA,Zeng,Weimann,Leykam,Xu,
Menke,Xiao,Lieu,Zyuzin,Fan2018PRB,HZhou18,Yin,Xiong,Shen,Kunst,Yao1,Yao2,Gong,Kawabata2,Takata,YChen,WYi2018,JHu2018,
HZhang2018,Rechtsman2018,Song2018,RYu2018,Kunst19,Xu2019,Sato2019,HZhou19,Kawabata1,Herviou}. Such systems
exist naturally or artificially due to gain or loss arising from the finite lifetime of quasiparticles~\cite{Fu2014}, the interaction with environment~\cite{Bender1, Bender2}, the engineered complex refractive index~\cite{Musslimani,Feng2017Nat}
and the engineered Laplacian in electric circuits~\cite{RYu2018,Thomale1}. A number of new topological phases have been found, such as
anomalous edge modes corresponding to half a winding number in a non-Hermitian Su-Schrieffer-Heeger model~\cite{Lee3}, Weyl exceptional rings with both quantized Chern number
and quantized Berry phase~\cite{Xu} and anomalous corner modes in non-Hermitian higher order topological insulators~\cite{Ueda2019PRL,Edvardsson2018,Luo2019}.

While there have been extensive studies of topological non-Hermitian phenomena including classification of non-Hermitian topological phases~\cite{Gong,Sato2019,HZhou19}, the one-dimensional (1D) Aubry-Andr\'e-Harper (AAH) model~\cite{Aubry, Harper} has been largely overlooked and not well explored. The AAH model, a 1D system modulated by an on-site cosinusoidal potential, plays a very important role in investigating the Anderson localization and topological phases~\cite{Siggia83,Kohmoto83,DasSarma88,DasSarma90,DasSarma09,DasSarma10,LJLang, Zilberberg, Zilberberg2012b, Ganeshan, Cai, Chong, Hu, Zeng2, Yi}.
Specifically, the model can be mapped to a two-dimensional (2D) Hall effect system with topological edge modes~\cite{LJLang,Zilberberg,Zilberberg2012b}.
Further generalization to an off-diagonal AAH model leads to a topological phase with zero-energy modes. Another very
interesting aspect is that this model gives rise to a topological quasicrystal when the incommensurate modulation is
considered~\cite{Zilberberg,Zilberberg2012b}.

In this paper, we study the topological phases in a non-Hermitian off-diagonal AAH model with a purely imaginary cosinusoidal modulation and asymmetric hopping under both commensurate and incommensurate scenarios. We find that (i) non-Hermitian
AAH models can host zero-energy modes at the edges. In contrast to the Hermitian counterpart, the zero-energy mode can be
localized only at one edge. Such a topological phase corresponds to the existence of the structure of energy bands
in momentum space enclosing a branch point of order 3~\cite{Needham}, in contrast to the previously discovered
structure enclosing a branch point of order 1 in the SSH model~\cite{Lee3}. That implies that starting at any quasimomentum $k=k_0$ corresponding to an energy $E_0$, we will return to this original energy $E(k_0)$ if we continuously
follow the value of the energy $E(k)$ as the quasimomentum varies from $k_0$ to $k_0+8\pi$. This leads to
a winding number being one quarter defined by the eigenenergy.
(ii) We further find the coexistence of a zero-energy mode located only at one edge
and nonzero energy edge modes. For the latter edge modes, we show that they can be characterized by a generalized Bott
index in a system under open boundary conditions (OBCs).
(iii) For incommensurate non-Hermitian quasicrystals, we demonstrate that both two edge modes and all bulk states are
localized at one edge, in stark contrast to the Hermitian case where all bulk states are extended and two edge modes are localized at two edges. Such topological edge modes can also be characterized by the generalized Bott index.
Finally, we propose an experimental scheme with electric circuits for realizing the
non-Hermitian AAH models.

\emph{Model Hamiltonian}.--- We start by considering the following 1D non-Hermitian AAH model
\begin{equation}\label{Hamiltonian}
  \hat{H} = \sum_{j} t(1 -\gamma+ \lambda_j) \hat{c}_{j+1}^\dagger \hat{c}_j +t(1 +\gamma+ \lambda_j) \hat{c}_{j}^\dagger \hat{c}_{j+1},
\end{equation}
where $\hat{c}_j^\dagger$ ($\hat{c}_j$) is the creation (annihilation) operator for a spinless particle at site $j$,
$t$ and $\gamma$ denote the hopping strength and an asymmetric hopping strength, respectively, and $\lambda_j = i \lambda \cos (2\pi \alpha j + \delta)$ depicts an imaginary modulation with $\lambda$, $\alpha$ and $\delta$ being real parameters.
When $\alpha$ is a rational number such that $\alpha=p/q$ with $p$ and $q$ being relatively prime positive integers, the modulation is periodic with $q$ being its period,
whereas the modulation becomes quasiperiodic, when $\alpha$ is an irrational number.

To determine the eigenenergy and eigenstates of the system under OBCs, we write the Hamiltonian
as $\hat{H}=\hat{c}^\dagger H\hat{c}$ where $\hat{c}=(\begin{array}{cccc}
                                                       \hat{c}_1 & \hat{c}_2 & \cdots & \hat{c}_L
                                                     \end{array})
$ with $L$ being the number of sites and diagonalize the Hamiltonian $H^\dagger$ and $H$ allowing us to obtain
both the left and right eigenstates $|\Psi_n^L\rangle$ and $|\Psi_n^R\rangle$ which satisfy
$H^\dagger|\Psi_n^L\rangle=E_n^*|\Psi_n^L\rangle$ and $H|\Psi_n^R\rangle=E_n|\Psi_n^R\rangle$
($E_n$ is the corresponding eigenenergy), respectively.
In the commensurate case, the Hamiltonian is translational invariant with respect to $q$ sites
under periodic boundary conditions. As a result, we can write the
Hamiltonian in momentum space as $\hat{H}=\sum_k \hat{c}^\dagger_k H(k)\hat{c}_k$ where
$\hat{c}_k=(\begin{array}{cccc}
                                                       \hat{c}_{1k} & e^{-ik/q}\hat{c}_{2k} & \cdots & e^{-i(q-1)k/q}\hat{c}_{qk}
\end{array})$ with $k \in [0, 2\pi]$ and $H(k)_{mn}=\delta_{mn-1}t_m +\delta_{m-1n}t_n^\prime+\delta_{m1}\delta_{nq}t_q^\prime e^{-ik}+\delta_{mq}\delta_{n1}t_qe^{ik}$ with $t_j=t(1+\gamma+\lambda_j)$ and $t_j^\prime=t(1-\gamma+\lambda_j)$.
Note that we have scaled the quasimomentum $k$ so that $k \in [0, 2\pi]$.
The left and right eigenvectors in momentum space $|\Psi_n^L(k)\rangle$ and $|\Psi_n^R(k)\rangle$
can be obtained by diagonalizing the matrix $H^\dagger(k)$ and $H(k)$, respectively.

\emph{Zero-energy modes in the commensurate AAH model}--- Let us first consider the commensurate modulation. To show the topological features, we first consider the simplest case with $\alpha=1/4$. In Fig.~\ref{fig1}(a),
we map out the topological phase diagram with respect to $\delta$ and $\gamma$, showing four distinct topological phases
characterized by $(W,N_e)$, where $W$ and $N_e$ denote the winding number of the Hamiltonian in momentum space and the number of
zero-energy edge eigenstates, respectively. These four phases correspond to $(W,N_e)=(-1,2),(-1/2,1),(-1/2,0),(0,0)$, which
will be elaborated on in the following discussion.

\begin{figure}[t]
  \includegraphics[width=3.3in]{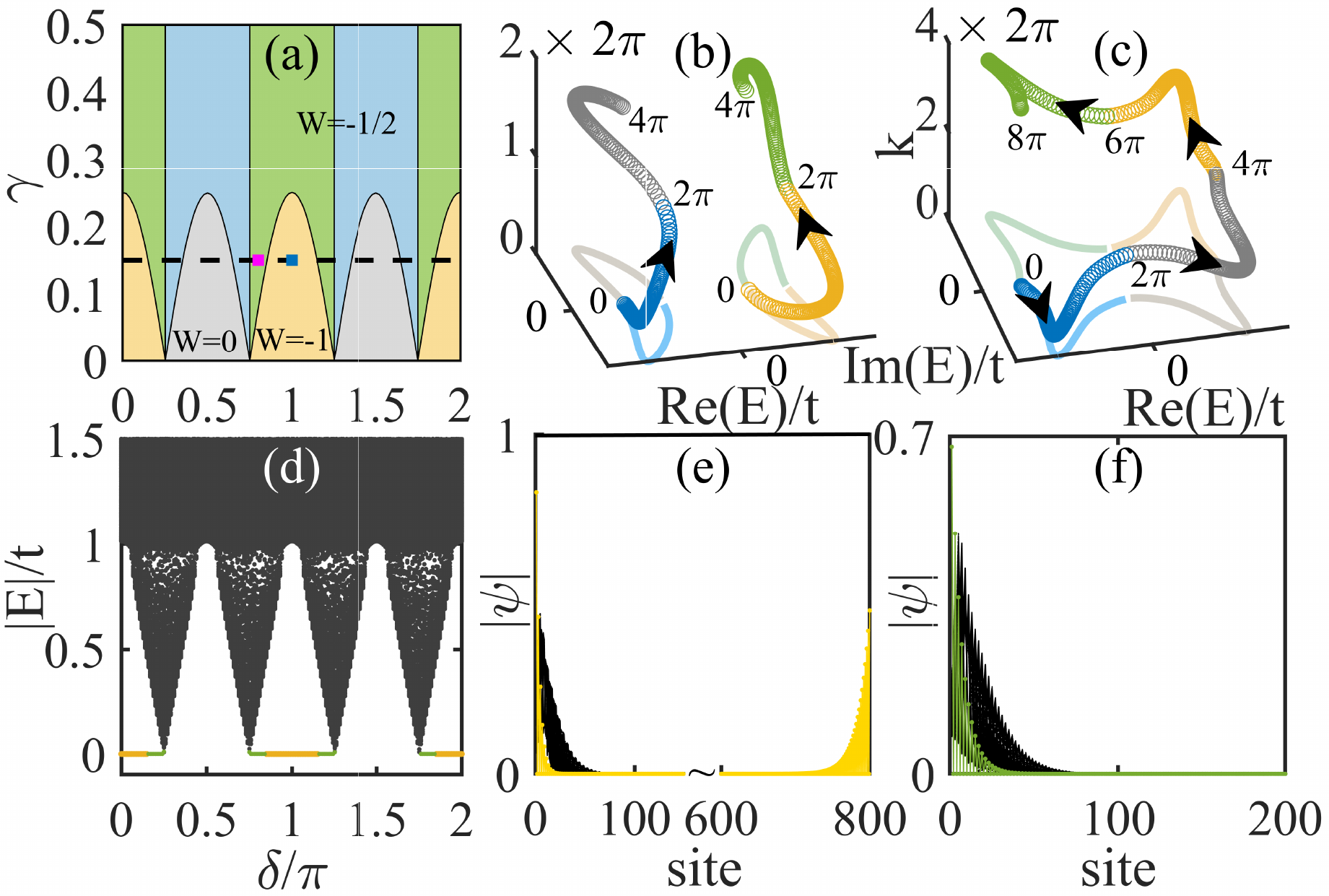}
  \caption{(Color online) (a) Phase diagram with respect to $\delta$ and $\gamma$, where the yellow, green, blue and
  gray areas represent the phases with $(W,N_e)=(-1,2),(-1/2,1),(-1/2,0),(0,0)$, respectively.
  In the region with $W=-1/2$, we find $W_E=1/4$. Complex energy spectra versus $k$ and their projection
  in the $k=0$ plane for (b) $\delta=\pi,\gamma=0.15$ and (c) $\delta=0.8\pi,\gamma=0.15$ corresponding to the blue and pink squares in (a), respectively. We also label the position of $k$ where $k=2m\pi$ with $m$ being an integer. (d) Absolute values of
  the eigenenergy versus $\delta$ under OBCs for $\gamma=0.15$ as indicated by the dashed black line in (a). The green and yellow lines at zero energy depict the one zero-energy eigenstate located only at the left edge
  and two zero-energy eigenstates located at both edges, respectively. (e-f) Amplitudes of wave functions
  with parameters indicated by the blue and pink squares in (a), respectively.
  The black lines denote the bulk states while the yellow and green lines the zero-energy edge states. Here, the lattice size $L=800$, $\lambda=1$, and $\alpha=1/4$. }
\label{fig1}
\end{figure}

Since the 1D system for a fixed $\delta$ respects the sublattice symmetry~\cite{SM,ChiuReview}, $H(k)$ can be transformed into
an off-diagonal block form~\cite{SM}:
$H(k)\rightarrow [0\text{  }h_1(k);h_2(k)\text{  }0]$, and the winding number for each block is defined as~\cite{Gong}
$
  w_{1,2}=\int_{0}^{2\pi} \frac{dk}{2\pi i} \partial_k \log \det h_{1,2}(k).
$
We can further define the winding number of the system as $W\equiv (w_1-w_2)/2$~\cite{Gong}.
In the Hermitian case, $h_2=h_1^\dagger$ leading to $w_1=-w_2$ and thus $W$ has to equal an integer.
However, the non-Hermitian term breaks this relation so that $W$ can be a half integer~\cite{Gong}.
This occurs in our system with $W=-1/2$ (see Fig.~\ref{fig1}(a)). For a system with two energy bands, such
as the SSH model, if $W=\pm 1/2$, we have $w_1=n$ and $w_2=n\pm1$ with $n$ being an integer and
thus $h_1\propto e^{in\theta_1(k)}$ and $h_2\propto e^{i(n\pm 1)\theta_2(k)}$, where $\theta_\nu(k)$ ($\nu=1,2$) changes continuously from $\theta_\nu(k_0)$ to $\theta_\nu(k_0)+2\pi$
as $k$ varies from $k_0$ to $k_0+2\pi$. Since the eigenenergy is $E_k=\pm\sqrt{-h_1(k)h_2(k)}
\propto e^{i n(\theta_1+\theta_2)/2}e^{i\theta_2/2}$,
implying that one ends up with the other energy $-E$ starting from one energy $E$ as $k$ varies
from $k_0$ to $k_0+2\pi$.

However, in our system, we find that when $W=-1/2$, all these four energy bands are connected (see Fig.~\ref{fig1}(c)),
implying that $E_k\propto e^{i\theta(k)/4}$, where $\theta(k)$ gains a $2\pi$ as $k$ continuously vary
from $k_0$ to $k_0+2\pi$, similar to $\theta_1$.
This shows that the energy encloses a branch point of order 3 so that a state needs to travel across the
Brillouin zone four times to return. To discriminate with the case involving a branch point of order 1,
we define a winding number for a separable energy band $E_n$ as
\begin{equation}
W_{E_n}=\frac{1}{2m\pi}\int_0^{2m\pi} dk \partial_k \text{arg}[E_n(k)-E_B]
\end{equation}
with respect to a base energy $E_B$, where $E_n(k)=E_n(k+2m\pi)$ with $m$ being the smallest integer so that
this relation is satisfied. For the non-Hermitian SSH model involving a branch point of order 1,
$W_E=1/2$. However, in our system when $W=-1/2$, we find $W_E=1/4$. Further calculation of the Berry phase
$C_1=\int_0^{2m\pi} dk \langle \Psi_n^L(k) |\partial_{k}\Psi_n^R(k)\rangle/\langle \Psi_n^L(k)|\Psi_n^R(k)\rangle $
as $k$ varies from $0$ to $8\pi$ shows that $C_1 \text{mod} 2\pi=\pi$~\cite{Xu}. Interestingly,
in the region with $W=-1$ and $W=0$ and $\gamma\neq 0$, we see that each separable energy bands encloses
a branch point of order 1, yielding $W_{E_1}=W_{E_2}=1/2$ with respect to the corresponding base energies
inside the rings (see Fig.~\ref{fig1}(b)).

Under OBCs, we show that when $W=-1$, there appear two zero-energy edge states
located at two edges as shown in Fig.~\ref{fig1}(e). While this is similar to the Hermitian case, different
properties arise that all bulk states are localized at the left edge when $\gamma>0$ due to the non-Hermitian
skin effects arising from the asymmetric hopping.
More interestingly, when $W=-1/2$, we find a region (green) where there is only one zero-energy
eigenstate located only at the left edge (see Fig.~\ref{fig1} (f)). In fact, the system exhibits a
zero-energy exceptional point with a zero-energy eigenstate and a zero-energy generalized eigenstate,
where the Hamiltonian becomes defective. This is also reflected
by the change of $N-\textrm{rank}(H)$ with $N=L$ from $2$ to $1$ as $\gamma$ varies from the yellow region
to the green one.

In addition, we see that there exists a region (blue) where despite $W=-1/2$ and $W_E=1/4$, no zero-energy modes emerge, implying the
breakdown of the bulk-edge correspondence (here bulk correspond to the wave functions in momentum space).
This arises from the dramatic change of the bulk wave functions
as boundary conditions are changed~\cite{Xiong}. To restore the bulk-edge correspondence, we need to use the wave
functions under OBCs to calculate the winding number. Let us follow the method
proposed in Ref.~\cite{Yao1} and calculate the $\text{det}(H(\beta)-EI)=0$ where $I$ is an identity
matrix and $H(\beta)=H(e^{ik}\rightarrow \beta)$ with the Hamiltonian $H$ in momentum space~\cite{SM}. This
equation gives us two solutions $\beta_1$ and $\beta_2$ for each $E$ satisfying
$\beta_1\beta_2=\prod_{j=1}^q t_j^\prime/\prod_{j=1}^q t_j$. For the bulk states, $|\beta_1|=|\beta_2|=r$.
This leads to a generalized Bloch Hamiltonian $\tilde{H}=H(e^{ik}\rightarrow r e^{ik})$ so that calculation of
the winding number of this Hamiltonian gives us the phase boundary for the existence of zero-energy modes.
In fact, this new Hamiltonian gives the same winding number as the case without $\gamma$.
For $\gamma=0$, we do not find any skin effects so that the bulk-edge correspondence is preserved,
implying that the gap closing of the energy bands in momentum space with respect to $\delta$ signals
whether zero-energy edge modes appear. We find that the gap closes when $\delta=(2j+1)\pi/4$ with $j=0,1,2,3$
and zero-energy edge modes emerge when $|\sin \delta| < |\cos \delta|$ as shown in Fig.~\ref{fig1}.

\begin{figure}[t]
  \includegraphics[width=3.3in]{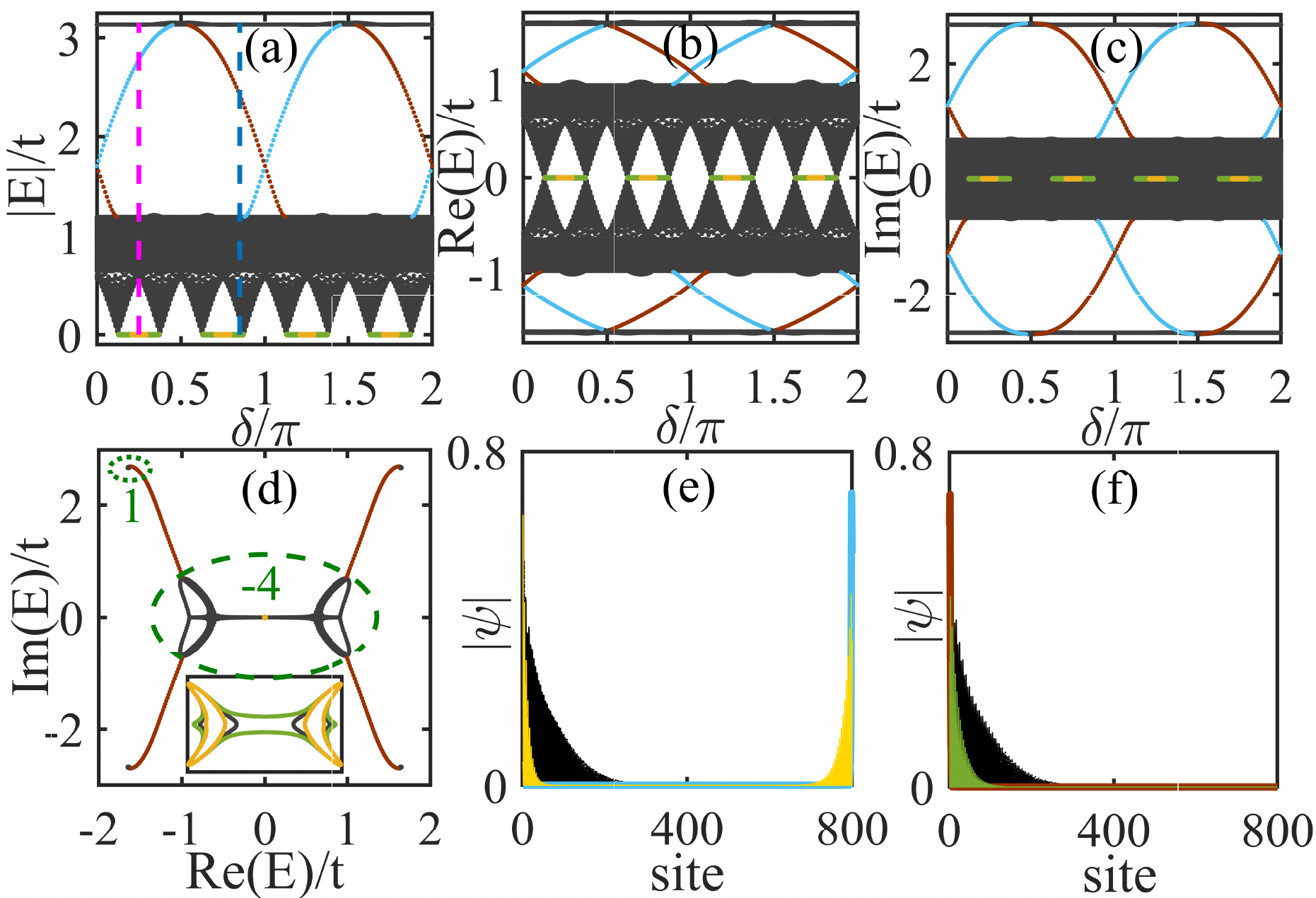}
  \caption{(Color online)
  (a) Absolute values, (b) real parts and (c) imaginary parts of complex eigenenergy versus $\delta$
  under OBCs. The green and yellow lines
  denote the one and two zero-energy edge states, respectively. The cyan and
  brown lines represent the nonzero energy edge states located at the left and right edges, respectively.
  (d) Energy spectra in the complex energy plane for all $\delta$ under OBCs. The brown lines
  represent the nonzero energy edge states located at two edges. The inset presents energy
  spectra of the Hamiltonian in momentum space for three different $\delta$ values: $0.05\pi$ (black), $0.12\pi$ (green) and $0.25\pi$ (yellow). The green numbers show the generalized Bott index for the states inside the corresponding circles.
  (e-f) Amplitudes of eigenstates for $\delta=\pi/4$ and $0.85\pi$ corresponding to the dashed pink and blue lines in (a), respectively. The black lines depict the bulk states. In (e), the yellow (cyan) lines
  correspond to two zero-energy edge states (nonzero energy edge states). In (d), the green (brown)
  line denotes the zero-energy (nonzero energy) edge mode. Here,
  $L=800$, $\lambda=2$, $\alpha=1/8$ and $\gamma=0.05$.}
  \label{fig2}
\end{figure}

In the general case, when $q=4m+2$ ($m$ being an integer and $4m$ being prime to $p$)
instead of a multiple of $4$, we find that
the energy spectrum of $\tilde{H}$ is gapless with the presence of zero-energy eigenstates for every $\delta \in [0,2\pi]$ \cite{SM}, indicating the absence of the zero-energy edge modes in such cases~\cite{footnote1}.
When $q=4m$, we have proved that the spectrum of $\tilde{H}$ is gapless when
$\delta = (2n+1) \pi/(4m)$ with $n = 0,1,\cdots,4m-1$ (suppose $m>0$)~\cite{SM}. When $\gamma=0$, it
is proved that a gapped region can appear,
showing that the topologically nontrivial zero-energy modes can exist~\cite{SM}.
In other cases, for instance, when $q$ is an odd number, there is no sublattice symmetry and thus the zero-energy states cannot be protected.

\emph{Coexistence of distinct types of edge modes in the commensurate AAH model}---
The non-Hermitian AAH model also exhibits a peculiar feature that the single zero-energy mode can coexist with other
topological nonzero energy edge modes (see Fig.~\ref{fig2}).
Specifically, Fig.~\ref{fig2} shows that
there exist two regions with one and two zero-energy edge states, respectively. In the former region,
$W_E=1/4$ for the eigenstates in momentum space.
Besides the zero-energy states,
we find other edge modes inside a gap, reminiscent
of chiral edge modes in a Chern insulator if $\delta$ is viewed as a quasimomentum. In the complex energy plane
for all $\delta$,
we observe five separable bands with four lines connecting four bands outside to one at the
center; these four lines
correspond to the edge states.

\begin{figure}[t]
  \includegraphics[width=3.4in]{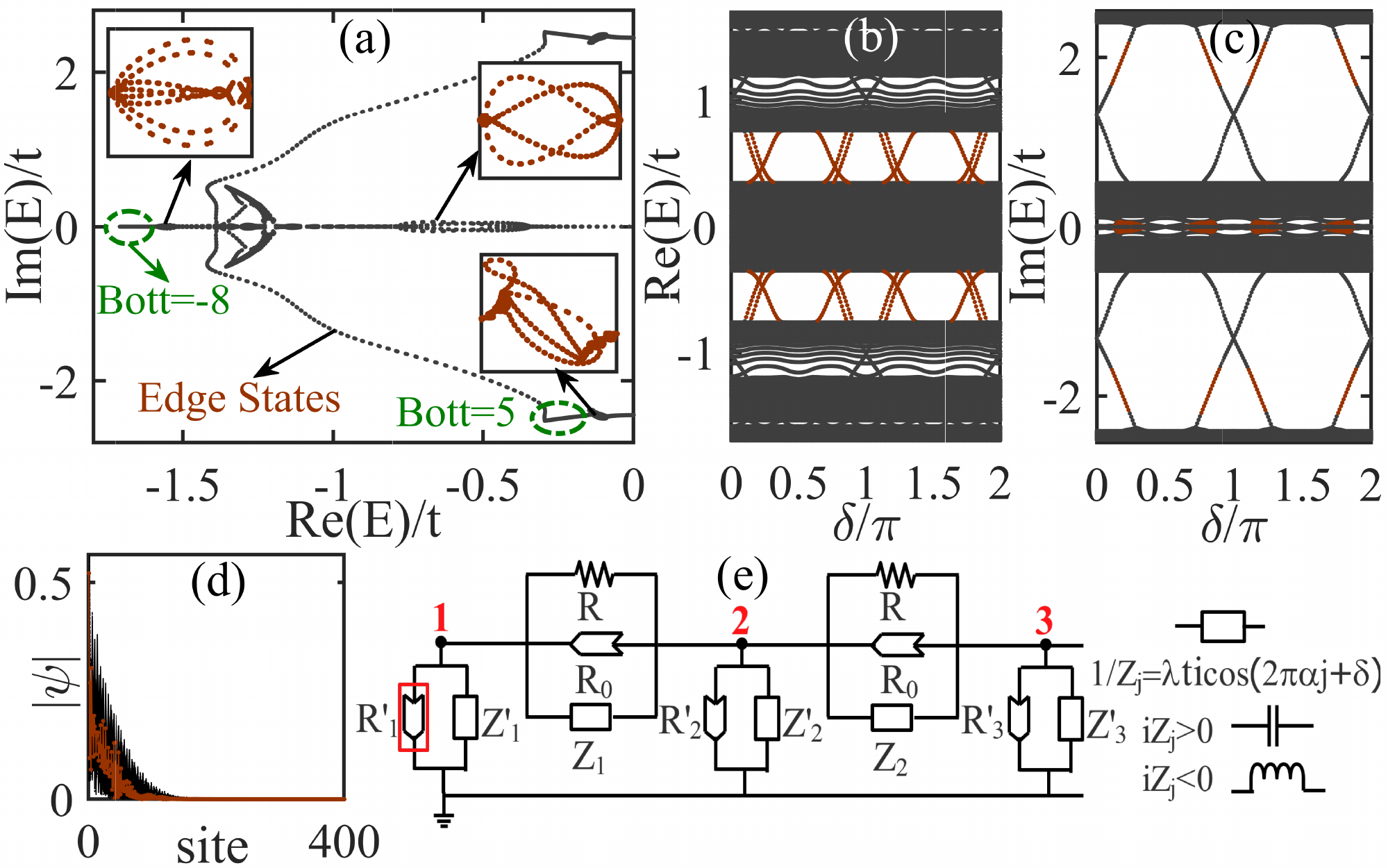}
  \caption{(Color online) (a) Energy spectra in the complex energy plane for the non-Hermitian quasicrystal.
  The insets plot the zoomed-in figures
  of the edge states inside the corresponding energy gaps. Real (b) and imaginary (c) parts of the complex energies.
  (d) Amplitudes of the bulk states (black lines) and edge states (brown lines) for $\delta=\pi/4$.
  For the figures above, $L=610$, $\lambda=2$, $\alpha=(\sqrt{5}-1)/2$, and $\gamma=0.1$.
  (e) Schematics of the electric circuit configuration for realizing the non-Hermitian AAH model.
  The electric element in the red box denotes a negative impedance converter with current inversion
  (INIC)~\cite{WKChen,Thomale3}. $R_j^\prime$ and $R_0$ denote the resistance of the INICs, the sign of which
  depends on the orientation of the INIC. $R$ represents the
  impedance of a resistor, and $Z_j$ and $Z_j^\prime$ that of a capacitor or a inductor. The rectangle electric element represents a capacitor or inductor determined by the value
  of the hopping as shown in the right figures.
}
  \label{fig3}
\end{figure}

When $\gamma > 0$, we find that all bulk states are localized at the left edge, implying that we need to
use the wave functions obtained under OBCs to characterize the ``chiral" edge states.
Here, we generalize the Bott index~\cite{Hastings} by defining it as
\begin{equation}
\mathrm{Bott}=\frac{1}{2\pi}\text{Im}\text{Tr}\log{U_y U_x U_y^\dagger U_x^\dagger},
\end{equation}
where $U_{a,mn}=\langle \Psi_m^L| e^{2\pi i\hat{a}/L_a}| \Psi_n^R\rangle$ with $a=x,y$
and $\hat{x}$ and $\hat{y}$ denoting the position operators along $x$ or $y$,
respectively, and $L_a$ labelling the size of the system along
the corresponding direction. Additionally, $| \Psi_n^R\rangle$ and $| \Psi_m^L\rangle$ represent
the right and left eigenvectors in a separable band, respectively. To calculate the Bott index, we map our system into a 2D Harper model~\cite{footnoteHaper}.
Transforming this Hamiltonian along $y$ to the form in momentum space exactly gives us
the Hamiltonian (\ref{Hamiltonian}) if $k_y$ is replaced with $\delta$. This allows us to calculate the Bott index
of $H_{2D}$ under periodic (open) boundary conditions along $y$ ($x$) to obtain
the topological invariant of our system.
We find that for the five separable bands, the Bott index is $-4$ for the central band and $1$ for each of the other four bands at the corners of the complex energy plane (see Fig.~\ref{fig2}(d)), demonstrating that the edge states
are topologically protected. We note that, with OBCs,
while there appear edge states connecting the separable bands, their presence does not
affect our results.

\emph{Non-Hermitian quasicrystals}.--- When $\alpha$ is irrational, the non-Hermitian AAH model becomes quasiperiodic, and the imaginary modulation is incommensurate with lattice spacings, leading to a quasicrystal.
Similar to the commensurate scenario, in Fig.~\ref{fig3}(a), we illustrate the energy spectrum in the complex energy plane for all $\delta$ from $0$ to $2\pi$ for $\alpha=(\sqrt{5}-1)/2$. The figure exhibits rich band structures.
Apparently, there are two separable bands
with the imaginary value around $\pm 2.5$. They are connected by the edge states (denoted by the brown lines) to the band
with real energies. For each of these two bands, there is also a mini-gap within which four edge states reside
(see the insets). For the band with real energies, there exist a gap and a mini-gap with four and eight edge states inside, respectively. These edge states can also be observed when the energy spectrum is projected to the real or imaginary part.

Remarkably, we further find that all bulk states are localized at the left edge when $\gamma\neq 0$ and all edge states
are located at the left edge when $\gamma$ is sufficiently large, in start contrast to the Hermitian case, as shown in
Fig.~\ref{fig3}(d). We note that the localization of the bulk states is caused by the non-Hermitian skin effect instead
of the Anderson localization. Since the bulk states are sensitive to the boundary conditions, we cannot
apply a twisted boundary condition to calculate the Chern number as the Hermitian case~\cite{Zilberberg}. Instead,
we can still calculate the Bott index using the wave functions obtained under OBCs. We find that the Bott index for each separable band equals the number of edge states inside the gap. For instance, the Bott index of the band with the imaginary value around $\pm 2.5$ and the real value smaller than $-0.1$ is 5, protecting
five edge states coming from this band (see Fig.~\ref{fig3}(a)).

\emph{Experimental realizations}.--- Recently, electric circuits have been demonstrated to be a powerful platform to simulate topological phenomena, such as the SSH model~\cite{Thomale1}, Weyl semimetals~\cite{Lu2019} and higher order topological insulators~\cite{Thomale2}.
Here, we propose an experimental scheme with electric circuits for realizing
the non-Hermitian AAH models (see Fig.~\ref{fig3}(e)). We can achieve the required the Laplacian so that
$J=-(EI+H)$ by
choosing appropriate impedances for these electric devices~\cite{footnote2}.
The edge states can be observed by measuring the two-point impedance between two nodes which diverges
as $E+E_n=0$ as we vary $E$.

In summary, we have demonstrated that for the commensurate non-Hermitian off-diagonal AAH model, there
exist zero-energy states localized at the edges. In contrast to the Hermitian case, the edge
states can be localized only at one edge. Such a topological phase corresponds to the emergence of a quarter
winding number defined by eigenenergy in momentum space. We further find that the zero-energy edge modes can coexist with nonzero energy edge modes protected by the generalized Bott index.
For the incommensurate case, topological non-Hermitian quasicrystals with edge modes are predicted. These
edge modes can be characterized by the generalized Bott index. Our findings pave the way for further studies
on topological properties in non-Hermitian Aubry-Andr\'e-Harper models.

\begin{acknowledgments}
\textit{Acknowledgement:} We thank S.-T. Wang for helpful discussions. This work was supported by the start-up program of Tsinghua University and the National Thousand-Young-Talents Program.
\end{acknowledgments}

\begin{widetext}

\section{Supplementary Materials}

\setcounter{equation}{0} \setcounter{figure}{0} \setcounter{table}{0} %
\renewcommand{\theequation}{S\arabic{equation}} \renewcommand{\thefigure}{S%
\arabic{figure}} \renewcommand{\bibnumfmt}[1]{[S#1]} \renewcommand{%
\citenumfont}[1]{S#1}
In the supplementary material, we will derive the generalized Bloch Hamiltonian for bulk states with open boundary conditions,
prove the condition under which the system is gapless around zero energy and discuss the condition for the
presence of the winding number of the Hamiltonian in detail.

\subsection{A. Generalized Bloch Hamiltonian}
We follow the method proposed in~\cite{WangZhong1,GeneralBloch} to obtain the generalized Hamiltonian.
For the commensurate case with $\alpha=p/q$ with $q$ and $p$ being mutually prime positive integers,
we can write the eigenstate of $H$ as $|\Psi_j^R \rangle=(\psi_{1,1}^j,\cdots,\psi_{1,q}^j,\cdots,\psi_{N,1}^j,\cdots,\psi_{N,q}^j)^T$
where $N$ is the number of unit cells. For open boundary conditions, let us suppose that $\psi_{n,\mu}^j=(\beta)^{n} \phi_{\mu}^{(j)}$.
The equation $H |\Psi_j^R \rangle=E |  \Psi_j^R \rangle$ leads to
\begin{equation}
H(\beta)
\left(
  \begin{array}{c}
    \phi^{(j)}_{1} \\
    \phi^{(j)}_{2} \\
    \phi^{(j)}_{3} \\
    \vdots \\
    \phi^{(j)}_{q}
  \end{array}
\right)
=
E
\left(
  \begin{array}{c}
    \phi^{(j)}_{1} \\
    \phi^{(j)}_{2} \\
    \phi^{(j)}_{3} \\
    \vdots \\
    \phi^{(j)}_{q}
  \end{array}
\right),
\end{equation}
where
\begin{equation}
H(\beta)=
\left(
  \begin{array}{ccccc}
    0 & t_1 & 0 & \cdots & t_q' \beta^{-1} \\
    t_1' & 0 & t_2 & \cdots & 0 \\
    0 & t_2' & 0 & \cdots & 0 \\
    \vdots & \vdots & \vdots & \ddots & \vdots \\
    t_q \beta & 0 & \cdots & \cdots & 0
  \end{array}
\right).
\end{equation}
To have a nontrivial solution, we require that $\det(EI-H(\beta))=0$. This gives us a quadratic equation
for $\beta(E)$ with two solutions $\beta_{1,2}$ satisfying
\begin{equation}
\beta_1\beta_2=\frac{t_1' t_2' \cdots t_q'}{t_1 t_2 \cdots t_q}.
\end{equation}
For bulk states, we require $|\beta_1|=|\beta_2|$~\cite{GeneralBloch} in order to obtain a continuum band. This gives us
\begin{equation}
|\beta_{1,2}|=r=\sqrt{\left|\frac{t_1' t_2' \cdots t_q'}{t_1 t_2 \cdots t_q}\right|}.
\end{equation}
The generalized Bloch Hamiltonian can be obtained by replacing the $e^{ik}$ with
$\beta=r e^{ik}$ in the Bloch Hamiltonian $H(k)$, that is,
\begin{equation}\label{H_beta}
\tilde{H}(k)=H(k\rightarrow k-i\log r)=
\left(
  \begin{array}{ccccc}
    0 & t_1 & 0 & \cdots & t_q' r^{-1} e^{-ik} \\
    t_1' & 0 & t_2 & \cdots & 0 \\
    0 & t_2' & 0 & \cdots & 0 \\
    \vdots & \vdots & \vdots & \ddots & \vdots \\
    t_q r e^{ik} & 0 & \cdots & \cdots & 0
  \end{array}
\right).
\end{equation}

In the following, we will use the generalized Bloch Hamiltonian $\tilde{H}(k)$ to determine the condition for
the existence of zero-energy edge states with open boundary conditions
and calculate the corresponding winding number.

\subsection{B. Condition for the existence of topological zero-energy modes}
In this section, we will show that for $q=4m+2$ with $m$ being an integer, the spectrum is
gapless around zero energy under periodic boundaries for all $\delta \in [0,2\pi)$,
while for $q=4m$,
the spectrum is gapless for $\delta = (2n+1) \pi/(4m)$ with $n = 0,1,\cdots,4m-1$ (suppose $m>0$).

When $q$ is an even number, the system has the sublattice symmetry $S_1^{-1} \tilde{H}(k)S_1=-\tilde{H}(k)$
with $S_1=\text{diag}(1,-1,1,-1,\cdots,1,-1)$ being a $q\times q$ diagonal matrix and we thus can transform $\tilde{H}(k)$ into the off-diagonal form: $\tilde{H}(k)=\left( \begin{array}{cc} 0 & \tilde{h}_1(k) \\ \tilde{h}_2(k) &0 \end{array} \right)$ with
\begin{equation}\label{hk}
  \tilde{h}_1(k)=\left(
  \begin{array}{ccccc}
    t_1 & 0 & 0 & \cdots & t_q' r^{-1} e^{-ik} \\
    t_2' & t_3 & 0 & \cdots & 0 \\
    0 & t_4' & t_5 & \cdots & 0 \\
    \vdots & \vdots & \vdots & \ddots & \vdots \\
    0 & 0 & 0 & \cdots & t_{q-1}
  \end{array}
  \right),
  \tilde{h}_2(k)=\left(
  \begin{array}{ccccc}
    t_1' & t_2 & 0 & \cdots & 0 \\
    0 & t_3' & t_4 & \cdots & 0 \\
    0 & 0 & t_5' & \cdots & 0 \\
    \vdots & \vdots & \vdots & \ddots & \vdots \\
    t_q r e^{ik} & 0 & 0 & \cdots & t_{q-1}'
  \end{array}
  \right),
\end{equation}
which are $\frac{q}{2} \times \frac{q}{2}$ matrices. Here $k\in[0,2\pi]$.
When $r=1$, $\tilde{h}_1=h_1$ and $\tilde{h}_2=h_2$.
If the determinant of $\tilde{H}(k)$ equals zero, i.e., $\det(\tilde{H}(k))=(-1)^{q^2/4} \det(\tilde{h}_1(k))\det(\tilde{h}_2(k))=0$, there will be eigenstates with zero eigenenergy.

In the case with $q=4m+2$, we have
\begin{align}
  \det(\tilde{h}_1(k))&=t_1 t_3 \cdots t_{4m+1} + t_2' t_4' \cdots t_{4m+2}' r^{-1} e^{-ik} \nonumber \\
  &=t_o + t_e' r^{-1} e^{-ik},
\end{align}
where $t_o=t_1 t_3 \cdots t_{4m+1}$ and $t_e'=t_2' t_4' \cdots t_{4m+2}'$.
This expression can be simplified to
\begin{equation}
\det(\tilde{h}_1(k))=t_e' r^{-1}(\Omega + e^{-ik}),
\end{equation}
where
\begin{eqnarray}
  \Omega &=& \frac{t_o r}{t_e'} = \frac{t_1 t_3 \cdots t_{4m+1}}{t_2' t_4' \cdots t_{4m+2}'} \sqrt{\left|\frac{t_1' t_2' \cdots t_{4m+2}'}{t_1 t_2 \cdots t_{4m+2}}\right|}, \\
  |\Omega| &=& \sqrt{\left|\frac{t_1 t_3 \cdots t_{4m+1}}{t_2 t_4 \cdots t_{4m+2}}\right| \left|\frac{t_1' t_3' \cdots t_{4m+1}'}{t_2' t_4' \cdots t_{4m+2}'}\right|}.
\end{eqnarray}
We have
\begin{align}
\left|\frac{t_1 t_3 \cdots t_{4m+1}}{t_2 t_4 \cdots t_{4m+2}}\right|
&=\left|\frac{\prod_{j=1}^{2m+1}[1+\gamma + i \lambda \cos(p \pi \frac{2j-1}{2m+1}+\delta)]}{\prod_{j=1}^{2m+1}[1+\gamma + i \lambda \cos(p \pi \frac{2j}{2m+1}+\delta)]}\right| \nonumber \\
&=\left|\frac{\prod_{j=1}^{2m+1}[1+\gamma - i \lambda \cos(p \pi \frac{2(j+m)}{2m+1}+\delta)]}{\prod_{j=1}^{2m+1}[1+\gamma + i \lambda \cos(p \pi \frac{2j}{2m+1}+\delta)]}\right| \nonumber \\
&=1.
\end{align}
Similarly, we obtain
\begin{equation}
\left|\frac{t_1' t_3' \cdots t_{4m+1}'}{t_2' t_4' \cdots t_{4m+2}'}\right|=1.
\end{equation}
Thus, we have $|\Omega|=1$, indicating that for each $\delta$, we can always find a $k_1\in[0,2\pi]$ such that $\det(\tilde{h}_1(k_1))=0$.
Therefore, we conclude that when $q=4m+2$, the system is gapless in the energy
spectrum for all $\delta$, implying the absence of the zero-energy edge states.

Before we consider the case for $q=4m$, we first present a lemma.
\begin{lemma}
Let $f_j,g_j\in \mathbb{R}$ with $j=1,\cdots,m$ and $m$ being an integer larger than zero. If $\prod_{j=1}^{m}(1+\lambda f_j)=\prod_{j=1}^{m}(1+\lambda g_j)$ for all $\lambda\in \mathbb{R}$, then for each $f_j$ with $1\leq j\leq m$, there exists a $g_r$ with $1\leq r\leq m$ such that $f_j=g_r$; conversely, for each $g_j$ with $1\leq j\leq m$,
there exists a $f_r$ with $1\leq r\leq m$ such that $g_j=f_r$.
\end{lemma}

\begin{proof}
Since the equation is satisfied for all $\lambda$, we have
\begin{eqnarray}
\sum_{j=1}^m f_j&=&\sum_{j=1}^m g_j=A_1, \label{Lm1} \\
\sum_{1\leq j_1<j_2\leq m}f_{j_1}f_{j_2}&=&\sum_{1\leq j_1<j_2\leq m}g_{j_1}g_{j_2}=A_2, \label{Lm2}\\
\sum_{1\leq j_1<j_2< j_3\leq m}f_{j_1}f_{j_2}f_{j_3}&=&\sum_{1\leq j_1<j_2<j_3\leq m}g_{j_1}g_{j_2}g_{j_3}=A_3, \label{Lm3}\\
&\cdots& \\
\prod_{j=1}^m f_j&=&\prod_{j=1}^m g_j=A_m. \label{Lm4}
\end{eqnarray}

Let $l$ be an integer such that $1\leq l\leq m$. With aids of Eq.~(\ref{Lm2}), multiplying Eq.~(\ref{Lm1}) by $f_l$ gives us
\begin{equation}
f_l^2-A_1 f_l+A_2=\sum_{1\leq j_1<j_2\leq m}\nolimits^\prime f_{j_1} f_{j_2}, \label{Lm5}
\end{equation}
where $\sum\nolimits^\prime$ indicates that its subscripts cannot be equal to $l$. We further multiply Eq.~(\ref{Lm5}) by
$f_l$ and, with aids of Eq.~(\ref{Lm3}), we obtain
\begin{equation}
f_l^3-A_1 f_l^2+A_2 f_l-A_3=-\sum\nolimits_{j_1<j_2<j_3}^\prime f_{j_1} f_{j_2} f_{j_3}.
\end{equation}
We repeat this process until we get
\begin{equation}
f_l^m-A_1 f_l^{n-1}+A_2 f_l^{n-2}-A_3 f_l^{n-3}+\cdots +(-1)^{m-1} A_{m-1} f_l+(-1)^m A_m=0.
\end{equation}
Since the left-hand expression can be written as $\prod_{j=1}^m (f_l-g_j)$, we have
\begin{equation}
\prod_{j=1}^m (f_l-g_j)=0.
\end{equation}
For all $l$, this equation holds, implying that, for each $f_l$, there exists a $g_r$ with $1\leq r \leq m$ such that
$f_l=g_r$. Conversely, similar derivation gives us
\begin{equation}
\prod_{j=1}^m (g_l-f_j)=0,
\end{equation}
implying that, for each $g_l$, there exists a $f_r$ with $1\leq r \leq m$ such that
$g_l=f_r$.
\end{proof}

When $q=4m$, we have
\begin{align}
  \det(\tilde{h}_1(k))&=t_1 t_3 \cdots t_{4m-1} - t_2' t_4' \cdots t_{4m}' r^{-1} e^{-ik} \nonumber \\
  &=t_o - t_e' r^{-1} e^{-ik} \nonumber \\
  &=t_e' r^{-1}(\Omega - e^{-ik}),
\end{align}
where $t_o=t_1 t_3 \cdots t_{4m-1}$, $t_e'=t_2' t_4' \cdots t_{4m}'$ and
\begin{eqnarray}
  \Omega &=& \frac{t_o r}{t_e'} = \frac{t_1 t_3 \cdots t_{4m-1}}{t_2' t_4' \cdots t_{4m}'} \sqrt{\left|\frac{t_1' t_2' \cdots t_{4m}'}{t_1 t_2 \cdots t_{4m}}\right|}, \\
  |\Omega| &=& \sqrt{T T'} = \sqrt{\left|\frac{t_1 t_3 \cdots t_{4m-1}}{t_2 t_4 \cdots t_{4m}}\right| \left|\frac{t_1' t_3' \cdots t_{4m-1}'}{t_2' t_4' \cdots t_{4m}'}\right|},
\end{eqnarray}
where
\begin{align}\label{omega_4m_1}
T\equiv\left|\frac{t_1 t_3 \cdots t_{4m-1}}{t_2 t_4 \cdots t_{4m}}\right|
&=\left|\frac{\prod_{j=1}^{2m}[1+\gamma + i \lambda \cos(p \pi \frac{2j-1}{2m}+\delta)]}{\prod_{j=1}^{2m}[1+\gamma + i \lambda \cos(p \pi \frac{2j}{2m}+\delta)]}\right| \nonumber \\
&=\frac{\prod_{j=1}^{m}[(1+\gamma)^2 + \lambda^2 \cos^2(p \pi \frac{2j-1}{2m}+\delta)]}{\prod_{j=1}^{m}[(1+\gamma)^2 + \lambda^2 \cos^2(p \pi \frac{2j}{2m}+\delta)]},
\end{align}
and
\begin{equation}
T'\equiv\left|\frac{t_1' t_3' \cdots t_{4m-1}'}{t_2' t_4' \cdots t_{4m}'}\right|
=\frac{\prod_{j=1}^{m}[(1-\gamma)^2 + \lambda^2 \cos^2(p \pi \frac{2j-1}{2m}+\delta)]}{\prod_{j=1}^{m}[(1-\gamma)^2 + \lambda^2 \cos^2(p \pi \frac{2j}{2m}+\delta)]}.
\end{equation}

When $\gamma=0$, we have
\begin{align}
|\Omega|=\frac{\prod_{j=1}^{m}[1 + \lambda^2 \cos^2(p \pi \frac{2j-1}{2m}+\delta)]}{\prod_{j=1}^{m}[1 + \lambda^2 \cos^2(p \pi \frac{2j}{2m}+\delta)]}.
\end{align}
Based on the lemma, if $|\Omega|=1$ for all $\lambda$,
we have
\begin{equation}
  \cos (p \pi \frac{2j_1-1}{m}+2 \delta) = \cos (p \pi \frac{2j_2}{m}+2 \delta),
\end{equation}
where $1\leq j_1,j_2\leq m$. This equation gives two types of possible solutions. For the first one,
\begin{equation}
(2j_1-1)p=2pj_2+2nm,
\end{equation}
with $n$ being an integer, which does not hold as odd numbers cannot be equal to even ones.
For the second one, we have
\begin{equation}
  p \pi \frac{2j_1-1}{m} + p \pi \frac{2j_2}{m} + 4 \delta = 2n \pi.
\end{equation}
Solving this equation shows that the gap of the energy spectrum closes when
\begin{equation}
  \delta = [ \frac{n}{2} - \frac{p}{4m} (2(j_1+j_2)-1) ] \pi,
\end{equation}
which is equivalent to
\begin{equation}
\delta = (2n+1) \pi/(4m)
\label{DeltaRelation}
\end{equation}
with $n = 0,1,\cdots,4m-1$. This tells us that the energy gap closes for all $\lambda$ when $\delta$ takes the above values.
When $\delta$ takes other values, there exist $\lambda$ so that the system is gapped, implying that the zero-energy edge
states can exist there.

When $\gamma \neq 0$, while we cannot prove that these $\delta$ in Eq.~\ref{DeltaRelation} are all the solutions
to $|\Omega|=1$ for all $\lambda$, we can verify that when $\delta$ take these values,
\begin{equation}
T=T'=1,
\end{equation}
yielding $|\Omega|=1$ and thus the energy gap closes at some $k$.

For each term $[(1+\gamma)^2 + \lambda^2 \cos^2(p \pi \frac{2j_1-1}{2m}+\delta)]$ in the numerator of $T$,
we can find a corresponding term $[(1+\gamma)^2 + \lambda^2 \cos^2(p \pi \frac{2j_2}{2m}+\delta)]$ in the denominator of $T$ to satisfy
\begin{equation}
  (1+\gamma)^2 + \lambda^2 \cos^2(p \pi \frac{2j_1-1}{2m}+\delta)=(1+\gamma)^2 + \lambda^2 \cos^2(p \pi \frac{2j_2}{2m}+\delta),
\end{equation}
if the two indices $j_1$ and $j_2$ satisfy
\begin{equation}
p \pi \frac{2j_1-1}{m} + p \pi \frac{2j_2}{m} + 4 \delta = 2n' \pi,
\end{equation}
with $n'$ being an integer.
This is true for $\delta = [ \frac{n'}{2} - \frac{p}{4m} (2(j_1+j_2)-1) ] \pi$ which is equivalent to $\delta = (2n+1) \pi/(4m)$ with $n = 0,1,\cdots,4m-1$.
It can also be seen that that $T'=1$ holds true in these cases.
Therefore, the energy spectrum is gapless at these $4m$ points.

\subsection{C. The Winding number of the generalized Bloch Hamiltonian}
Since the system respect the sublattice symmetry, we can use the winding number as the topological invariant
to characterize the zero-energy edge modes.
When $q=4m$, the winding number of the generalized Bloch Hamiltonian for each block is defined as
\begin{equation}
  w_{1,2}=\int_{0}^{2\pi} \frac{dk}{2\pi i} \partial_k \log \det \tilde{h}_{1,2}(k).
\end{equation}
We obtain
\begin{align}
  w_1 &= -w_2 \nonumber \\
      &=\int_{0}^{2\pi} \frac{dk}{2\pi i} \partial_k \log (\Omega - e^{-ik}) \nonumber \\
      &=\int_{0}^{2\pi} \frac{dk}{2\pi} \frac{e^{-ik}}{\Omega-e^{-ik}} \\
      &=\frac{i}{2\pi}\oint_\Gamma \frac{dz}{\Omega-z},
\end{align}
where $\Gamma$ denotes an integral path along a clockwise unit circle in the complex plane.
Clearly, we have
$w_1=-1$ if $|\Omega|<1$ and $w_1=0$, if $|\Omega| >1$, corresponding to topologically
nontrivial and trivial regions, respectively.

\end{widetext}


\begin{thebibliography}{}
\bibitem{Hasan}
M. Z. Hasan and C. L. Kane, Rev. Mod. Phys. \textbf{82}, 3045 (2010).
\bibitem{Qi}
X.-L. Qi and S.-C. Zhang, Rev. Mod. Phys. \textbf{83}, 1057 (2011).
\bibitem{Armitage}
N. P. Armitage, E. J. Mele, and Ashvin Vishwanath, Rev. Mod. Phys. \textbf{90}, 015001 (2018).

\bibitem{Rudner1}
M. S. Rudner and L. S. Levitov, Phys. Rev. Lett. \textbf{102}, 065703 (2009).

\bibitem{Esaki}
K. Esaki, M. Sato, K. Hasebe, and M. Kohmoto, Phys. Rev. B \textbf{84}, 205128 (2011).

\bibitem{Bardyn}
C.-E. Bardyn, M. A. Baranov, C. V. Kraus, E. Rico, A. \.Imamo\u glu, P. Zoller, and S. Diehl, New J. Phys. \textbf{15}, 085001 (2013).

\bibitem{Poshakinskiy}
A. V. Poshakinskiy, A. N. Poddubny, L. Pilozzi, and E. L. Ivchenko, Phys. Rev. Lett. \textbf{112}, 107403 (2014).

\bibitem{Zeuner}
J. M. Zeuner, M. C. Rechtsman, Y. Plotnik, Y. Lumer, S. Nolte, M. S. Rudner, M. Segev, and A. Szameit, Phys. Rev. Lett. \textbf{115}, 040402 (2015).

\bibitem{Malzard}
S. Malzard, C. Poli, and H. Schomerus, Phys. Rev. Lett. \textbf{115}, 200402 (2015).

\bibitem{Rudner2}
M. S. Rudner, M. Levin, and L. S. Levitov, arXiv:1605.07652 (2016).

\bibitem{Aguado2016SR}
P. San-Jose, J. Cayao, E. Prada, and R. Aguado, Sci. Rep. \textbf{6,} 21427 (2016).

\bibitem{Lee3}
T. E. Lee, Phys. Rev. Lett. \textbf{116}, 133903 (2016).

\bibitem{Molina} J. Gonz\' alez and R. A. Molina, Phys. Rev. Lett. \textbf{116,} 156803 (2016).

\bibitem{Joglekar2016PRA}
A. K. Harter, T. E. Lee, and Y. N. Joglekar, Phys. Rev. A \textbf{93,} 062101 (2016).

\bibitem{Zeng}
Q. B. Zeng, B. Zhu, S. Chen, L. You, and R. L\"u, Phys. Rev. A \textbf{94}, 022119 (2016).

\bibitem{Weimann}
S. Weimann, M. Kremer, Y. Plotnik, Y. Lumer, S. Nolte, K. G. Makris, M. Segev, M. C. Rechtsman, and A. Szameit, Nat. Mater. \textbf{16}, 433 (2017).

\bibitem{Leykam}
D. Leykam, K. Y. Bliokh, C. Huang, Y. D. Chong, and F. Nori, Phys. Rev. Lett. \textbf{118}, 040401 (2017).

\bibitem{Xu}
Y. Xu, S.-T. Wang, and L.-M. Duan, Phys. Rev. Lett. \textbf{118}, 045701 (2017).

\bibitem{Menke}
H. Menke and M. M. Hirschmann, Phys. Rev. B \textbf{95}, 174506 (2017).

\bibitem{Xiao}
L. Xiao, X. Zhan, Z. H. Bian, K. K. Wang, X. Zhang, X. P. Wang, J. Li, K. Mochizuki, D. Kim, N. Kawakami, W. Yi, H. Obuse, B. C. Sanders, and P. Xue, Nat. Phys. \textbf{13}, 1117 (2017).

\bibitem{Lieu}
S. Lieu, Phys. Rev. B \textbf{97}, 045106 (2018).
\bibitem{Zyuzin}
A. A. Zyuzin and A. Y. Zyuzin, Phys. Rev. B \textbf{97}, 041203(R) (2018).

\bibitem{Fan2018PRB}
A. Cerjan, M. Xiao, L. Yuan, and S. Fan, Phys. Rev. B \textbf{97,} 075128 (2018).

\bibitem{HZhou18}
H. Zhou, C. Peng, Y. Yoon, C. W. Hsu, K. A. Nelson, L. Fu, J. D. Joannopoulos, M. Soljaci\'c, and B. Zhen, Science \textbf{359}, 1009 (2018).

\bibitem{Yin}
C. Yin, H. Jiang, L. Li, R. L\"u, and S. Chen, Phys. Rev. A \textbf{97}, 052115 (2018).

\bibitem{Xiong}
Y. Xiong, J. Phys. Commun. \textbf{2}, 035043 (2018).

\bibitem{Shen}
H. Shen, B. Zhen, and L. Fu, Phys. Rev. Lett. \textbf{120},146402 (2018).

\bibitem{Kunst}
F. K. Kunst, E. Edvardsson, J. C. Budich, and E. J. Bergholtz, Phys. Rev. Lett. \textbf{121}, 026808 (2018).

\bibitem{Yao1}
S. Yao and Z. Wang, Phys. Rev. Lett. \textbf{121}, 086803 (2018).
\bibitem{Yao2}
S. Yao F. Song, and Z. Wang, Phys. Rev. Lett. \textbf{121}, 136802 (2018).
\bibitem{Gong}
Z. Gong, Y. Ashida, K. Kawabata, K. Takasan, S. Higashikawa, and M. Ueda, Phys. Rev. X \textbf{8}, 031079 (2018).
\bibitem{Kawabata2}
K. Kawabata, K. Shiozaki, and M. Ueda, Phys. Rev. B \textbf{98}, 165148 (2018).

\bibitem{Takata}
K. Takata and M. Notomi, Phys. Rev. Lett. \textbf{121}, 213902 (2018).

\bibitem{YChen}
Y. Chen and H. Zhai, Phys. Rev. B \textbf{98}, 245130 (2018).

\bibitem{WYi2018}
X. Qiu, T.-S. Deng, Y. Hu, P. Xue, and W. Yi, arXiv:1806.10268 (2018).

\bibitem{JHu2018}
Z. Yang and J. Hu, arXiv: 1807.05661 (2018).

\bibitem{HZhang2018}
H. Wang, J. Ruan, and H. Zhang, arXiv: 1808.06162 (2018).

\bibitem{Rechtsman2018}
A. Cerjan, S. Huang, K. P. Chen, Y. Chong, and M. C. Rechtsman, arXiv: 1808.09541 (2018).

\bibitem{Song2018}
L. Jin and Z. Song, arXiv: 1809.03139 (2018).

\bibitem{RYu2018}
K. F. Luo, J. J. Feng, Y. X. Zhao, and R. Yu, arXiv:1810.09231 (2018).

\bibitem{Kunst19}
F. K. Kunst and V. Dwivedi, arXiv:1812.02186 (2018).

\bibitem{Xu2019}
Y. Xu, arXiv:1812.03756 (2018).

\bibitem{Sato2019}
K. Kawabata, K. Shiozaki, M. Ueda, and M. Sato, arXiv:1812.09133 (2018).

\bibitem{HZhou19}
H. Zhou and J. Y. Lee, arXiv:1812.10490 (2018).

\bibitem{Kawabata1}
K. Kawabata, S. Higashikawa, Z. Gong, Y. Ashida, and M. Ueda, Nat. Commun. \textbf{10}, 297 (2019).

\bibitem{Herviou}
L. Herviou, J. H. Bardarson, N. Regnault, arXiv:1901.00010 (2019).

\bibitem{Fu2014}
V. Kozii and L. Fu, arXiv:1708.05841 (2017).
\bibitem{Bender1}
C. M. Bender and S. Boettcher, Phys. Rev. Lett. \textbf{80}, 5243 (1998).
\bibitem{Bender2}
C. M. Bender, Rep. Prog. Phys. \textbf{70}, 947 (2007).
\bibitem{Musslimani}
Z. H. Musslimani, K. G. Makris, R. El-Ganainy, and D. N. Christodoulides, Phys. Rev. Lett. \textbf{100}, 030402 (2008).
\bibitem{Feng2017Nat}{L. Feng, R. El-Ganainy, and L. Ge, Nat. Photonics \textbf{11,} 752 (2017).}

\bibitem{Thomale1}
C. H. Lee, S. Imhof, C. Berger, F. Bayer, J. Brehm, L. W. Molenkamp, T. Kiessling, and R. Thomale, Commun. Phys.
\textbf{1,} 39 (2018).

\bibitem{Ueda2019PRL}
T. Liu, Y.-R. Zhang, Q. Ai, Z. Gong, K. Kawabata, M. Ueda, and F. Nori,
Phys. Rev. Lett. \textbf{122,} 076801 (2019).

\bibitem{Edvardsson2018}
E. Edvardsson, F. K. Kunst, and E. J. Bergholtz, arXiv:1812.09060.
\bibitem{Luo2019}
X.-W Luo and C. Zhang, arXiv:1903.02448.

\bibitem{Aubry}
S. Aubry and G. Andr\'e, Ann. Isr. Phys. Soc. \textbf{3}, 133 (1980).
\bibitem{Harper}
P. G. Harper, Proc. Phys. Soc. London Sect. A \textbf{68}, 874 (1955).

\bibitem{Siggia83}
S. Ostlund, R. Pandit, D. Rand, H. J. Schellnhuber, and E. D. Siggia, Phys. Rev. Lett. \textbf{50}, 1873 (1983).
\bibitem{Kohmoto83}
M. Kohmoto, Phys. Rev. Lett. \textbf{51}, 1198 (1983).
\bibitem{DasSarma88}
S. Das Sarma, S. He, and X. C. Xie, Phys. Rev. Lett. \textbf{61}, 2144 (1988).
\bibitem{DasSarma90}
S. Das Sarma, S. He, and X. C. Xie, Phys. Rev. B \textbf{41}, 5544 (1990).
\bibitem{DasSarma09}
J. Biddle, B. Wang, D. J. Priour Jr., and S. Das Sarma, Phys. Rev. A \textbf{80}, 021603(R) (2009).
\bibitem{DasSarma10}
J. Biddle and S. Das Sarma, Phys. Rev. Lett. \textbf{104}, 070601 (2010).

\bibitem{LJLang}
L.-J. Lang, X. Cai, and S. Chen, Phys. Rev. Lett. \textbf{108}, 220401 (2012).

\bibitem{Zilberberg}
Y. E. Kraus, Y. Lahini, Z. Ringel, M. Verbin, and O. Zilberberg, Phys. Rev. Lett. \textbf{109}, 106402 (2012).
\bibitem{Zilberberg2012b}
{Y. E. Kraus and O. Zilberberg, Phys. Rev. Lett. \textbf{109,} 116404 (2012).}
\bibitem{Ganeshan}
S. Ganeshan, K. Sun, and S. Das Sarma, Phys. Rev. Lett. \textbf{110}, 180403 (2013).
\bibitem{Cai}
X. Cai, L.-J. Lang, S. Chen, and Y. Wang, Phys. Rev. Lett. \textbf{110}, 176403 (2013).
\bibitem{Chong}
F. Liu, S. Ghosh, and Y. D. Chong, Phys. Rev. B \textbf{91}, 014108 (2015).
\bibitem{Hu}
J. Wang, X.-J. Liu, G. Xianlong, and H. Hu, Phys. Rev. B \textbf{93}, 104504 (2016).
\bibitem{Zeng2}
Q. B. Zeng, S. Chen, and R. L\"u, Phys. Rev. B \textbf{94}, 125408 (2016).
\bibitem{Yi}
C. M. Dai, W. Wang, and X. X. Yi, Phys. Rev. A \textbf{98}, 013635 (2018).

\bibitem{Needham}
T. Needham, \emph{Visual Complex analysis} (Oxford University Press, USA, 1999).

\bibitem{ChiuReview}
C.-K. Chiu, J. C. Y. Teo, A. P. Schnyder, and S. Ryu, Rev. Mod. Phys. \textbf{88}, 035005 (2016).

\bibitem{SM}
Supplementary Materials.

\bibitem{footnote1}{Here the generalized Bloch Hamiltonian $\tilde{H}$ is used instead of $H(k)$ because of the breakdown of the bulk-edge correspondence when $\gamma\neq 0$.}

\bibitem{Hastings}
T. A. Loring and M. B. Hastings, Europhys. Lett. \textbf{92,} 67004 (2010).

\bibitem{WKChen}
W.-K. Chen, \emph{The Circuits and Filters Handbook}, 3rd ed. (CRC Press, Inc., Boca Raton, FL, USA, 2009).
\bibitem{Thomale3}
T. Hofmann, T. Helbig, C. H. Lee, M. Greiter, and R. Thomale, arXiv:1809.08687.

\bibitem{footnoteHaper}{$H_{2D}= \sum_{j_x,j_y} t(1+\gamma) \hat{c}_{j_x,j_y}^\dagger \hat{c}_{j_x+1,j_y} + t(1-\gamma) \hat{c}_{j_x+1,j_y}^\dagger \hat{c}_{j_x,j_y}
+t\frac{i\lambda}{2} [ e^{-i 2\pi\alpha j_x} (\hat{c}_{j_x,j_y+1}^\dagger \hat{c}_{j_x+1,j_y} + \hat{c}_{j_x+1,j_y+1}^\dagger \hat{c}_{j_x,j_y})
+H.c.]$.}

\bibitem{Lu2019}
Y. Lu, N, Jia, L. Su, C. Owens, G. Juzeliunas, D. I. Schuster, and J. Simon, Phys. Rev. B \textbf{99}, 020302(R) (2019).
\bibitem{Thomale2}
S. Imhof, C. Berger, F. Bayer, J. Brehm, L. W. Molenkamp, T. Kiessling, F. Schindler, C. H. Lee, M. Greiter, T. Neupert, and R. Thomale, Nat. Phys. \textbf{14}, 925 (2018).

\bibitem{footnote2}{$1/R=t$, $1/R_0=t\gamma$, $1/Z_j=\lambda t i\cos(2\pi\alpha j+\delta)$,
$1/R_1^\prime=1/R_L^\prime =t(1-\gamma)+\text{Re}(E)$ and $1/Z_1^\prime=-1/Z_1-\text{Im}(E)$,
$1/Z_L^\prime=-1/Z_{L-1}-\text{Im}(E)$, and $1/R_j^\prime=2t+\text{Re}(E)$ and
$1/Z_j^\prime=-1/Z_{j-1}-1/Z_{j+1}-\text{Im}(E)$ when $1<j<L$.}

\end{thebibliography}

\begin{thebibliography}{99}
\bibitem{WangZhong1}{S. Yao and Z. Wang, Phys. Rev. Lett. \textbf{121}, 086803 (2018).}

\bibitem{GeneralBloch}{K. Yokomizo and S. Murakami, arXiv:1902.10958 (2019).}

\end{thebibliography}
\end{document}